\documentclass{article}

\usepackage[utf8]{inputenc} 
\usepackage[T1]{fontenc}    
\usepackage[colorlinks=true,linkcolor=black,anchorcolor=black,citecolor=black,filecolor=black,menucolor=black,runcolor=black,urlcolor=black]{hyperref}       
\usepackage{url}            
\usepackage{booktabs}       
\usepackage{amsfonts}       
\usepackage{nicefrac}       
\usepackage{microtype}     
\usepackage{xcolor}         
\usepackage{bm}         
\usepackage{graphicx}  		
\usepackage{amsmath} 
\usepackage{cite} 
\usepackage{diagbox}
\usepackage{subfig}
\usepackage{siunitx}
\usepackage{multirow}
\usepackage{tabularx}
\usepackage{hhline}
\usepackage{arydshln}		
\usepackage{xcolor}
\usepackage{mathtools} 
\usepackage{amsthm}			
\usepackage{amssymb}			
\usepackage{algorithm}
\usepackage{pifont}
\usepackage{threeparttable}
\usepackage{fullpage}
\usepackage[letterpaper,
            left=1.67cm,
            right=1.67cm,
            top=1.8cm,
            bottom=1.69in]{geometry}

\usepackage{wrapfig}

\usepackage{algpseudocode}	
\usepackage{arydshln}
\usepackage[export]{adjustbox}
\usepackage{booktabs}

\usepackage{mathtools} 
\usepackage{microtype} 
\usepackage{multirow}
\usepackage{nicefrac}       
\usepackage{pifont}   

\usepackage{url}
\usepackage{siunitx}
\usepackage{tabularx}
\usepackage{threeparttable}
\usepackage{xcolor}  

\definecolor{darkblue}{rgb}{0,0 ,0.542}
\definecolor{lightgreen}{rgb}{.9,1,.9}
\definecolor{lightred}{rgb}{1,.415,.415}
\definecolor{lightblue}{rgb}{.415,.415,1}

\newcolumntype{L}[1]{>{\raggedright\arraybackslash}p{#1}}
\newcolumntype{C}[1]{>{\centering\arraybackslash}p{#1}}
\newcolumntype{R}[1]{>{\raggedleft\arraybackslash}p{#1}}

\theoremstyle{plain} 

\newtheorem{theorem}{Theorem}
\newtheorem{lemma}{Lemma}

\newtheorem{assumption}{Assumption}

\newtheorem*{suptheorem}{Theorem}

\def\defn{\,\coloneqq\,}

\def\prox{{\mathsf{prox}}}

\def\min{\mathop{\mathsf{min}}}

\def\R{\mathbb{R}}
\def\E{\mathbb{E}}

\def\ebm{{\bm{e}}}

\def\xbm{{\bm{x}}}
\def\zbm{{\bm{z}}}
\def\ybm{{\bm{y}}}
\def\zbm{{\bm{z}}}

\def\nbm{{\bm{n}}}
\def\ubm{{\bm{u}}}
\def\vbm{{\bm{v}}}

\def\varepsilonbar{{\overline{\varepsilon}}}
\def\dsf{{\mathsf{d}}}

\def\Dsfhat{{\widehat{\mathsf{D}}}}

\def\Ibm{{\bm{I}}}

\def\thetabm{{\bm{\theta }}}

\def\Ibm{{\bm{I}}}

\def\Ncal{{\mathcal{N}}}

\def\Lcal{{\mathcal{L}}}

\def\Dsf{{\mathsf{D}}}

\def\Isf{{\mathsf{I}}}
\def\Rsf{{\mathsf{R}}}
\def\Tsf{{\mathsf{T}}}

\def\Tsf{{\mathsf{T}}}

\def\Dsf{{\mathsf{D}}}

\def\Isf{{\mathsf{I}}}

\def\Dcal{{\mathcal{D}}}

\def\xbmast{{\bm{x}^\ast}}
\def\xbmbar{{\overline{\bm{x}}}}

\def\xbmhat{{\widehat{\bm{x}}}}

\def\Dhat{{\mathsf{\widehat{D}}}}

\def\Rsfhat{{\mathsf{\widehat{R}}}}

\def\argmin{\mathop{\mathsf{arg\,min}}} 

\title{PnP Restoration  with Domain Adaptation for SANS}
\date{}

\author{Shirin Shoushtari, Edward P. Chandler, Jialiang Zhang, Manjula Senanayake,\\ Sai Venkatesh Pingali, Marcus Foston, and Ulugbek S. Kamilov.\\ 
\small Washington University in St. Louis, MO 63130 USA\\
\small Oak Ridge National Laboratory, Oak Ridge, Tennessee 37831, USA}

\begin{document}

\maketitle

\let\thefootnote\relax\footnote{This material is based upon work supported by the Gordon and Betty Moore Foundation grant 11396.}

\vspace{-3em} 
\begin{abstract}
\medskip\noindent
Small Angle Neutron Scattering (SANS) is a non-destructive technique utilized to probe the nano- to mesoscale structure of materials by analyzing the scattering pattern of neutrons. Accelerating SANS acquisition for in-situ analysis is essential, but it often reduces the signal-to-noise ratio (SNR), highlighting the need for methods to enhance SNR even with short acquisition times. While deep learning (DL) can be used for enhancing SNR of low quality SANS, the amount of experimental data available for training is usually severely limited. We address this issue by proposing a \textbf{P}lug-and-play \textbf{R}estoration for \textbf{SANS} (\textbf{PR-SANS}) that uses domain-adapted priors. The prior in PR-SANS is initially trained on a set of generic images and subsequently fine-tuned using a limited amount of experimental SANS data. We present a theoretical convergence analysis of PR-SANS by focusing on the error resulting from using inexact domain-adapted priors instead of the ideal ones.
We demonstrate with experimentally collected SANS data that PR-SANS can recover high-SNR 2D SANS detector images from low-SNR detector images, effectively increasing the SNR. This advancement enables a reduction in acquisition times by a factor of 12 while maintaining the original signal quality.
\end{abstract}

\section{Introduction}
Small angle neutron scattering (SANS) is widely used technique for investigating material structure at mesoscopic scale~\cite{grillo200813}.
Long exposure to the neutron beam is often necessary for detailed structural information from the SANS experiments. 
There has been an increased demand for insights into the dynamic behaviors of mesoscale structures, motivating the development of in-situ and operando SANS. Traditionally, a limitation of SANS has been its lower neutron flux compared to x-ray scattering, leading to longer acquisition times and diminished temporal resolution. Recent advancements have seen the introduction of SANS instruments equipped with larger and faster detectors, enhanced guide coating technologies, and high-flux neutron sources, aiming to overcome these challenges~\cite{hollamby2013practical, hainbuchner2000new}. Despite these improvements, obtaining detailed structural information often requires significant signal averaging and prolonged acquisitions times, complicating high-quality data acquisition for applications like in-situ SANS. Consequently, attempts to expedite data collection frequently compromise the temporal resolution or SNR of the measurements~\cite{granroth2018event}. 

The recovery of high-SNR SANS detector images from low-SNR SANS measurements can be formulated as an
image restoration problem. Deep learning (DL) has been extensively used as a data-driven strategy to address a wide range of imaging inverse problems~\cite{McCann.etal2017, Lucas.etal2018}. Traditional DL approaches are based on training a neural network to map low-quality observations to high-quality images~\cite{Metzler.etal2018, Zhang.etal2017a, Meinhardt.etal2017, zafari2023frequency,Dong.etal2019}. There is an increasing interest in leveraging DL for enhancing the quality of SANS data~\cite{chang2020deep}, predicting structural geometry of sample material~\cite{mutti2019deep}, and performing data analysis and feature extraction from SANS~\cite{do2020small, doucet2020machine}. 
Despite the promising performance of DL in computationally enhancing quality of SANS data~\cite{chang2020deep}, the lack of experimentally collected SANS data poses a challenge for training effective DL models. Domain adaptation is a promising strategy to address the issue of limited data availability for training DL models, a challenge often faced in SANS experiments where acquiring large-scale datasets is impractical. Domain adaptation involves pre-training a DL model on one dataset (domain) and adapting it to perform effectively on a different domain~\cite{gopalan2011domain,tommasi2012improving, Shoushtari.etal2023}.

Plug-and-play (PnP) priors~\cite{Venkatakrishnan.etal2013, Sreehari.etal2016} is a computational framework widely-used for designing DL methods for solving imaging inverse problems. PnP methods leverage powerful DL denoisers as priors by combining them with the forward model of the imaging instruments. PnP has been successfully used in various imaging inverse problems including super-resolution, phase retrieval, microscopy, and medical imaging~\cite{Metzler.etal2018, Meinhardt.etal2017, Zhang.etal2021b, liu2022recovery} (see recent reviews~\cite{Ahmad.etal2020, kamilov2023plug}).

Despite the success of PnP on a number of imaging inverse problems, it has never been applied for restoring high-SNR SANS data. In this paper, we present the first investigation into PnP for restoring high-SNR SANS data collected in an accelerated fashion that results in low-SNR SANS detector images. Our PnP restoration method, called PR-SANS, restores high-SNR detector images from low-SNR measurements with shorter acquisition time, thereby mitigating the need for long acquisition times in SANS experiments. To address the limited amount of experimentally collected SANS datasets for training deep priors for PnP, we propose a domain adaptation strategy that adapts a generic image prior (trained on natural grayscale images) into a dedicated SANS prior. We present a theoretical convergence analysis of PR-SANS that considers the errors introduced by using adapted deep priors instead of ideal ones. We show on experimentally collected SANS data that PR-SANS can restore high-SNR SANS detector images while significantly reducing the acquisition times for detector image of equivalent SNR ($12\times$ faster data acquisition). Our results numerically show the relationship between the amount of training SANS data and the final restoration performance. 


\begin{figure}[!t]
    \centering
    \includegraphics[width=1\textwidth]{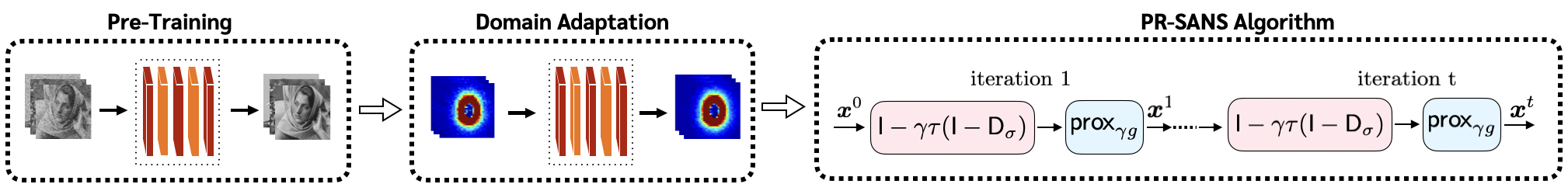}
    \caption{Illustration of the proposed method: Initially, a deep prior is pre-trained on a dataset of natural grayscale images. It is then adapted to the target dataset (SANS detector images) using a limited number of samples. Once adapted, this prior is plugged into the PR-SANS algorithm to effectively restore the 2D SANS detector images. }
    \label{fig:shematic}
\end{figure}

\section{Proposed Method}
\textbf{Image restoration.} We propose to formulate the recovery of the clean image $\xbm \in \R^n$ from the noisy SANS measurement $\ybm \in \R^n$ as an image restoration problem of form
\begin{equation}\label{eq:linearSys}
\ybm = \xbm + \ebm, 
\end{equation}
where $\ebm \in \R^n$ is the noise. It is common to formulate restoration problems as an   optimization problem
\begin{equation}
    \label{eq:inverseproblem}
    \xbmhat = \argmin_{\xbm \in \R^n} f(\xbm) \quad \text{with} \quad f(\xbm) = g(\xbm) + h(\xbm), 
\end{equation}
where $g$ denotes the data-fidelity term that measures the consistency of recovered signal with the measurement $\ybm$ and $h$ is a regularizer that imposes prior knowledge on the recovered signal. Widely-used data-fidelity and regularization terms are least-squares  
$g(\xbm) = \frac{1}{2}\|\xbm - \ybm\|_2^2$ and total variation (TV) regularizer $h(\xbm) = \tau \|\bm{D}(\xbm)\|_1$, where $\bm{D}$ is image gradient and $\tau>0$ is the regularization parameter~\cite{Rudin.etal1992}.

\medskip
\textbf{PnP.} PnP is one of the most popular approaches for solving imaging inverse problems where image priors are specified using pre-trained image denoisers~\cite{kamilov2023plug}. Our proposed PR-SANS is a variant of PnP with the update rule
\begin{equation}
    \label{eq:updaterule}
    \xbm^k = \prox_{\gamma g}(\xbm^{k-1} - \gamma \tau \Rsfhat(\xbm^{k-1}))\quad \text{with} \quad \Rsfhat = \Isf - \Dsfhat_\sigma,
\end{equation}
where $\Dsfhat_\sigma$ is the domain-adapted deep prior with parameter $\sigma >0$ for controlling the noise strength, $\tau>0$ is the regularization parameter, and $\gamma>0$ is the penalty parameter. We denote the ideal SANS prior as mapping $\Dsf_\sigma : \R^n \to \R^n$ pre-trained to solve the following restoration problem
\begin{equation}
    \label{eq:noisemodel}
    \zbm = \xbm + \nbm \quad \text{with} \quad \nbm \sim \Ncal(0, \sigma^2 \Ibm), 
\end{equation}
where $\nbm \in \R^n$ is additive white Gaussian noise (AWGN) of variance $\sigma^2$~\cite{bigdeli2017,Reehorst.Schniter2019}. 
As common in PnP methods~\cite{kamilov2023plug}, we assume that $\Dsf_\sigma$ performs MMSE estimation of $\xbm \in \R^n$ for the problem~\eqref{eq:noisemodel}
\begin{equation}\label{eq:mmsedenoiser}
    \Dsf_{\sigma}(\zbm) = \E[\xbm|\zbm] = \int_{\R^n} \xbm p_{\xbm|\zbm}(\xbm;\zbm) \dsf \xbm, 
\end{equation}
where $p_\xbm$ represent the  distribution of target dataset.
When the PnP method converges, it converges to $\xbmast \in \R^n$ that satisfies 
\begin{equation}
    \label{eq:fixedpoint}
    \nabla g(\xbmast) + \tau \Rsf(\xbmast) = \bm{0},
\end{equation}
where $\xbmast \in \R^n$ is the fixed point. 
We discuss the connection of fixed points with stationary points of an objective function $f = g+h$ for some regularizer $h$ in Section~\ref{sec:app}. 

\medskip
\textbf{Domain adaptation.}
In traditional DL frameworks for addressing image restoration tasks, it is common
to train convolutional neural networks (CNNs) to map observations to desired images. In this approach, the training and testing data are often drawn from the same distribution~\cite{Ongie.etal2020,Wang.etal2016, Jin.etal2017, Kang.etal2017}. 
However, domain shift can occur when training and testing data originate from different distributions~\cite{Jalal.etal2021a, guan2021domain, Darestani.etal2021, sun2020test, shoushtari2022deep}. 
Domain adaptation techniques can address this issue by generalizing a DL model trained on a source domain to a domain of interest (target domain)~\cite{tommasi2012improving, tommasi2013learning, Tirer.Giryes2019}. Domain adaptation is particularly beneficial in scenarios  where the available training data is insufficient~\cite{farahani2021brief }. Among adaptation techniques, a simple strategy involves fine-tuning a pre-trained neural network using limited available data from the target domain. This method relies on the premise that image features can be transferred across different domains~\cite{farahani2021brief,gopalan2011domain,Shoushtari.etal2023 }. 

Our domain adaptation relies on a dataset of natural images $\Dcal_s = \{\xbm_i^s,\zbm_i^s \}_{i=1}^N$ as the source dataset. $\Dcal_s$ is used to train the initial denoiser $\Dsf^s_\sigma$ by minimizing the MSE loss 
\begin{equation}
\label{Eq:MSELoss}
\Lcal(\Dsf^s_\sigma) = \E \left[ \|\xbm^s - \Dsf^s_{\sigma}(\zbm^s)\|_2^2 \right]. 
\end{equation}
The parameters of this initial denoiser, denoted by $\thetabm^s$
, are stored for the subsequent adaptation phase. The target dataset  $\Dcal_t = \{\xbm_i^t,\zbm_i^t \}_{i=1}^K$ (where $K \ll N$) includes a small number of SANS detector images. 
The adaptation involves fine-tuning the initial denoiser using the target dataset $\Dcal_t$ with the MSE loss
\begin{equation}
\label{Eq:MSELoss2}
\Lcal(\Dhat_\sigma) = \E \left[ \|\xbm^t - \Dhat_{\sigma}(\zbm^t)\|_2^2 \right],
\end{equation}
starting from the initial parameters $\thetabm^s$ of the denoiser pre-trained on source data.

\section{Theoretical Analysis}
\noindent
In this section, we present the theoretical convergence of the iterates generated by PR-SANS with adapted priors. Our analysis will require several assumptions that act as sufficient conditions for our theoretical results.

\begin{assumption}\label{as:nondegen}
    The prior density $p_{\xbm}$ is non-degenerate over $\R^n$. 
\end{assumption}
If the support of the probability density $p_{\xbm}$ is confined to space with a lower dimension than $n$, it is degenerate over $\R^n$. We use the assumption of non-degenerative prior density to establish a link between the operator~\eqref{eq:mmsedenoiser} and the following regularizer 
\begin{equation}\label{eq:regden}
    h(\xbm) = -\tau \sigma^2 \log p_{\zbm}(\xbm), \quad \xbm \in \R^n,     
\end{equation}
where $\tau$ is the regularization parameter, $p_\zbm$ is the density of observation~\eqref{eq:noisemodel}, and $\sigma^2$ is the AWGN level used for training $\Dsf_\sigma$ (the derivation is in Lemma~\ref{Lem:regularizerDenRel}).
The link between MMSE operators and and regularization beyond non-degenerate priors can be found in~\cite{Gribonval.Machart2013}. It can be proven that function $h$ is infinitely continuously differentiable~\cite{Gribonval2011, Gribonval.Machart2013}.
\begin{assumption}
\label{As:InexactDistance}
The adapted restoration operator $\Dhat_{\sigma}$ satisfies
\begin{equation*}
\|\Dhat_{\sigma}(\xbm^k)-\Dsf_{\sigma}(\xbm^k)\|_2 \leq \varepsilon_k, \quad k = 0,1, 2, 3, \ldots
\end{equation*}
where $\Dsf_{\sigma}$ is given in~\eqref{eq:mmsedenoiser}.
\end{assumption}
This assumption serves as a bounding constraint on the discrepancy between the domain adapted $\Dsfhat_\sigma$ and ideal MMSE restoration operator $\Dsf_\sigma$ for SANS data at each iteration of PR-SANS algorithm. 

Our analysis assumes that at every iteration, PR-SANS uses the adapted MMSE restoration operator, where the ideal MMSE restoration operator $\Dsf_\sigma$ is unavailable. We consider the case where at iteration $k$ of PR-SANS, the distance of the outputs of $\Dsf_{\sigma}$ and $\Dhat_{\sigma}$ is bounded by a constant $\varepsilon_k$.
\begin{assumption}
    \label{As:lipofG}
    The function $g$ is continuously differentiable. Additionally, $\nabla g$ and $\nabla h$ are Lipschitz continuous with constants $L>0$ and $M>0$, respectively. 
\end{assumption}
Lipschitz continuity is a standard assumption in the context of imaging inverse problems~\cite{Hurault.etal2022}.
\begin{assumption}
    \label{As:boundednessF}
    The data-fidelity term $g$ and the implicit regularizer $h$ are bounded from below. 
\end{assumption}
This assumption implies that there exists $f^\ast = f(\xbm^*) > -\infty$ such that $f(\xbm)\geq f^\ast$ for all $\xbm \in \R^n$.
\begin{theorem}
\label{Thm:maintheorem}
Run PR-SANS with the adapted MMSE restoration operator $\Dhat_\sigma$ for $t\geq 1$ iterations under Assumptions~\ref{as:nondegen}-\ref{As:boundednessF} using a step-size $\gamma >0$. Then, for each iteration $1\leq k\leq t$, we have 
\begin{equation*}
    \min_{1\leq k \leq t} \|\nabla f(\xbm^k)\|^2_2 \leq \frac{1}{t}\sum_{k=1}^t \|\nabla f(\xbm^k)\|^2_2 
    \leq \frac{B_1}{t}\left(f(\xbm^0) - f(\xbm^*)\right) + B_2 \varepsilonbar^2_t, 
\end{equation*}
where $B_1>0$ and $B_2>0$ are iteration independent constants, and $\varepsilonbar_t^2 \defn  (1/t)\sum_{k=1}^t\varepsilon^2_{k-1}$ . Additionally, if the sequence of error terms $\{\varepsilon_i\}_{i \geq 0}$ is square-summable, we have that $\|\nabla f(\xbm^k)\| \to 0$ as $k\to \infty$.
\end{theorem}

The expressions for constant $B_1$ and $B_2$ are given in the proof. Theorem 1 shows that the iterates generated by PR-SANS with adapted priors satisfy $\nabla f(\xbm^k) \to \bm{0}$ as $k\to \infty$ when the sequence of error term $\varepsilon_k$ is square-summable. The theorem shows that if the sequence of error is square-summable, PR-SANS asymptotically achieves a stationary point of $f$. Additionally, if the sequence of error is \emph{not} square-summable, one can still precisely characterize the error due to the discrepancy between adapted and ideal MMSE denoisers.   

\section{Experimental Results}

We perform numerical validation of PR-SANS for restoring experimentally collected SANS detector images. Our experiments are designed to evaluate the effectiveness of the PR-SANS algorithm by comparing its performance against baselines, such as the total variation (TV) denoising algorithm and end-to-end CNNs. Furthermore, we highlight how the complexity of the sample in the domain adaptation of deep priors influences the performance of PR-SANS.

\begin{wrapfigure}{r}{9cm}
\includegraphics[width=0.47\textwidth]{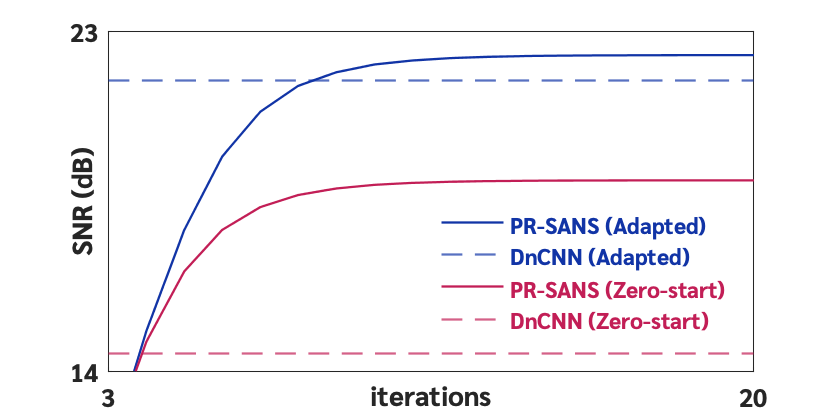}
    \caption{{Empirical evaluation of PR-SANS convergence to the true solution $\xbmast$ with adapted and zero-start DL priors and end-to-end models (DnCNN). The average SNR is plotted against iterations of the PR-SANS algorithm. Note that the proposed algorithm outperforms the DL model. Also, note the performance enhancement due to using adapted models instead of zero-start ones. }}
    \label{fig:psnr}
\end{wrapfigure}

We compare PR-SANS performance with several established  methods including end-to-end CNNs and TV denoising algorithm, which represents a non-learning-based approach to noise reduction~\cite{Beck.Teboulle2009a}. DnCNN~\cite{Zhang.etal2017}, a specific type of CNN, serves dual functions in our analysis: (a) it acts as an end-to-end network that directly maps measurements to detector images, and (b) it functions as a deep image prior in the PR-SANS algorithm. DnCNN models that are trained with a limited amount of data are labeled as "zero-start," while those trained with domain adaptation strategy are referred to as "adapted". Additionally, to investigate the effect of sample complexity on the domain adaptation of deep image priors, we evaluate the performance of PR-SANS with multiple adapted priors that differ in the quantity of target domain samples used for adaptation. Our evaluation metrics include SNR and root-square-mean-error (RMSE) for the recovery of 2D detector images, alongside normalized-mean-square-error (NMSE) and mean absolute error (MAE) for performance analysis of 1D intensity plots. We used a pre-trained DnCNN model, initially trained for reducing AWGN with a noise level parameter of $\sigma = 5$ on 400 natural grayscale images~\cite{Martin.etal2001}. For the training of both adapted and zero-start priors, 100 pairs of low- and high-SNR SANS detector images were used, with an extra set of 10 SANS detector image pairs reserved for validation and 10 pairs reserved for testing. The models were trained for 300 epoch and the model with the best performance on validation set was chosen. The number of iteration for PR-SANS was set to 20. The penalty parameter is set to $\gamma = 0.7$ for all the experiments. The regularization parameter $\tau$ was optimized for the best performance. The data collection procedure for SANS experiments is detailed in section~\ref{app:datacollection}.
\begin{figure*}[!t]
    \centering
    \includegraphics[width=1\textwidth]{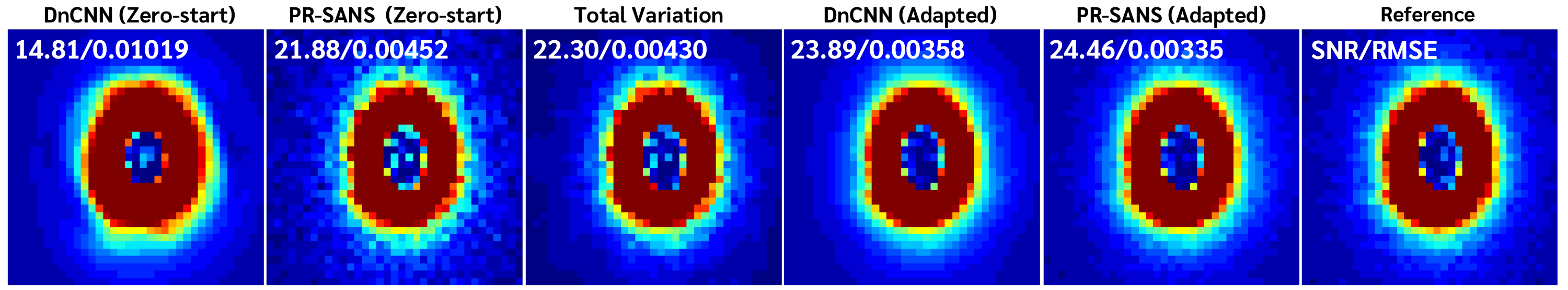}
    \caption{Visual evaluation of several methods on restoring high-SNR SANS detector images from low SNR data. Zero-start refers to the models that are trained with limited experimental data and adapted refers to models with domain adaptation. Note that generally, adapted models have better performance. Also note the improvement acquired by using the proposed restoration algorithm in comparison with just using DL models. }
    \label{fig:recon2d}
\end{figure*}
\begin{figure*}[!t]
    \centering    \includegraphics[width=1\textwidth]{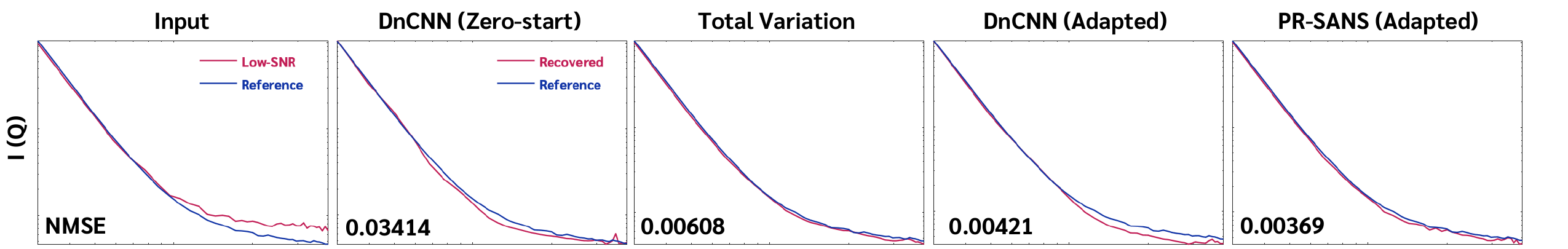}
    \caption{Visual evaluation of several methods on recovering 1D intensity $\bf{I (Q)}$ plot against scattering vector $\bf{Q(\AA)}$  from low SNR SANS data. Note the competitive performance of the proposed method against other approaches.}
    \label{fig:recon1d}
\end{figure*}
\subsection{SANS detector image restoration}\label{ssec:firstsetexp}
In this section, we compare the performance of PR-SANS using adapted and zero-start priors against zero-start DnCNN, adapted DnCNN, and TV denoising algorithm. 
Figure~\ref{fig:recon2d} illustrates the visual comparison of various methods,  evaluating their performance in terms of SNR and RMSE. Note the superior performance of PR-SANS when using an adapted DnCNN prior, as opposed to the adapted DnCNN functioning as an end-to-end mapping. Additionally, note the performance improvement due to domain adaptation in PR-SANS algorithm and DnCNN. Figure~\ref{fig:recon1d} illustrates visual comparison of 1D intensity $I (Q)$ plot against scattering vector $Q(\AA)$ (Section~\ref{app:datacollection} provides details on acquiring  $I (Q)$ from 2D detector images). Note that the proposed method PR-SANS with adapted prior achieves the best performance. 
Table \ref{tab:methodscomparison} presents the quantitative evaluation of 2D detector images and 1D intensity plots restoration performance, showing that PR-SANS outperforms other methods.
Figure~\ref{fig:psnr} illustrates the convergence of the PR-SANS with adapted and zero-start priors.
\begin{figure*}[!t]
    \centering 
    \includegraphics[width=0.95\textwidth]{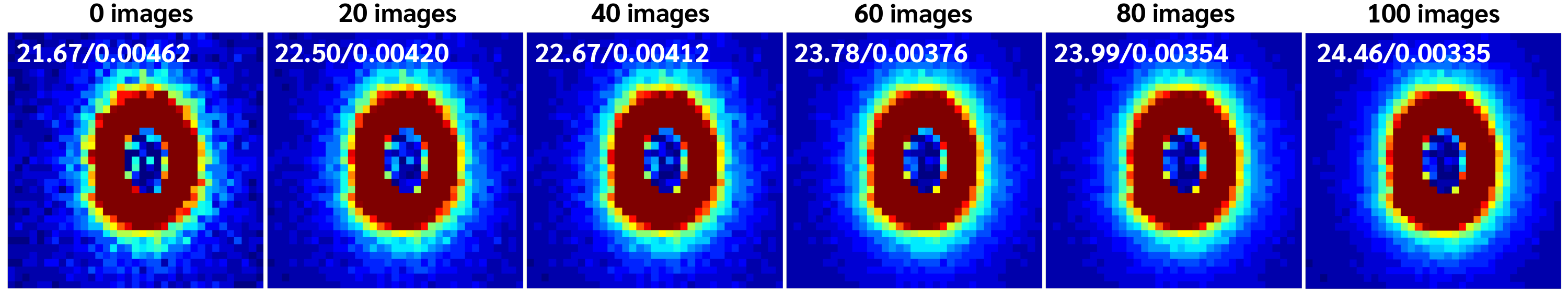}
    \caption{Visual evaluation of sample complexity used in domain adaptation of DL prior for PR-SANS algorithm. The performance is reported in terms of SNR (dB) and RMSE. Note the improvement of performance when the adaptation samples from the SANS dataset increases.}
    \label{fig:reconadap}
\end{figure*}
\begin{table}[!t]
 \caption{Numerical evaluation of 2D detector image and 1D intensity plot restoration. Table highlights the {\color{lightred}{best}} and {\color{lightblue}{second best}} results. Note the suboptimal performance of DnCNN (Zero-start) due to the limited amount of training data, which highlights the importance of domain adaptation.}
\centering\setlength\tabcolsep{4pt}\renewcommand{\arraystretch}{1}
    \begin{tabular}{ l l  l l l l l}
    \midrule
          \textbf{Method} & \multicolumn{2}{c}{\textbf{2D}} && \multicolumn{2}{c}{\textbf{1D}}\\
          \cline{2-3}\cline{5-7}
          &\textbf{SNR} $\uparrow$&\textbf{RMSE} $\downarrow$&&\textbf{NMSE} $\downarrow$&\textbf{MAE} $\downarrow$\\
          \midrule
          \textbf{Low-SNR  data}&$18.44$  &$0.00680$ && $0.03683$&$ 0.00316$\\
          \hdashline
          \textbf{DnCNN (Zero-start)}&$14.48$  &$0.01058$ && $0.03683$&$0.00230$\\ 
          \textbf{PR-SANS (Zero-start)}& $19.05$ &$0.00634$ &&$0.01358$ &$0.00204$\\
          \textbf{Total Variation}& $19.53$ & $0.00600$&& $0.01219$&$0.00258$\\
          \textbf{DnCNN (Adapted)} & \color{lightblue}$\bf{21.69}$ &\color{lightblue}$\bf{0.00468}$ &&\color{lightblue}$\bf{0.00748}$ &\color{lightblue}$\bf{0.00162}$\\
          \textbf{PR-SANS (Adapted)}& \color{lightred}$\bf{22.36}$ &\color{lightred}$\bf{0.00432}$ &&\color{lightred}$\bf{0.00635}$ &\color{lightred}$\bf{0.00144}$\\
          \midrule
    \end{tabular}
    \label{tab:methodscomparison}
\end{table}

\subsection{Sample 
Complexity of Domain Adaptation}\label{ssec:domainadaptexp}
In this section, we investigate the effect of sample complexity used for prior adaptation on PR-SANS performance, particularly focusing on how the number of detector images selected from the target domain (SANS) for prior adaptation influences outcomes. 
Same pre-trained model is adapted using different number of randomly chosen data pair from SANS dataset.

Figure~\ref{fig:reconadap} illustrates the visual comparison of  2D detector images restoration using PR-SANS with different adapted priors . The performance is reported using SNR (dB) and RMSE. In this context, ``0 images'' refers to PR-SANS with the pre-trained base DnCNN without any domain adaptation. Note the progressive improvement in the PR-SANS algorithm's performance as the sample complexity used in domain adaptation increases. Numerical results for the performance of adapted models in recovering 2D detector image and 1D intensity plot, averaged for all test samples, are reported in Table~\ref{tab:adaptation}.

Figure~\ref{fig:psnrofSamples} presents the empirical results comparing the performance PR-SANS with adapted and zero-start priors. The performance in terms of SNR(dB) is plotted against the number of samples used for adaptation and training of zero-start priors.

Theorem~\ref{Thm:maintheorem} establishes that PR-SANS approximates the stationary points of objective function with an error margin, where the error depends on the difference between the ideal $\Dsf_\sigma$ and adapted $\Dsfhat_\sigma$ priors (see Assumption~\ref{As:InexactDistance}). 
Obtaining the ideal prior for SANS images is not feasible in practice.
Nonetheless, it can be  shown that increasing the number of images used for domain adaptation can gradually bring the adapted $\Dsfhat_\sigma$ closer to the ideal $\Dsf_\sigma$ prior, decreasing the restoration error margin. 
Figure~\ref{fig:error} illustrates the averaged error of 2D detector image restoration against PR-SANS iterations for various priors. Note the reduction in error as the number of images used for adaptation grows, a result that is consistent with the result in Theorem~\ref{Thm:maintheorem}.

\begin{table}[!t]
 \caption{Numerical evaluation of sample complexity for domain adaptation of DL prior in PR-SANS algorithm. Table highlights the {\color{lightred} best} results. Note how performance relates to the number of samples used for adaptation.}
\centering\setlength\tabcolsep{4pt}\renewcommand{\arraystretch}{1}
    \begin{tabular}{ l  l l l l l l}
    \midrule
          \textbf{Method} & \multicolumn{2}{c}{\textbf{2D}} && \multicolumn{2}{c}{\textbf{1D}}\\
          \cline{2-3}\cline{5-7}
          &\textbf{SNR} $\uparrow$&\textbf{RMSE} $\downarrow$&&\textbf{NMSE} $\downarrow$&\textbf{MAE} $\downarrow$\\
          \midrule
          \textbf{Adapted (0 imgs)}&$18.46$  &$0.00679$ && $0.01566$&$0.00316$\\ 
          \textbf{Adapted (20 imgs)}& $20.17$ &$0.00556$ &&$0.01047$ &$0.00228$\\
          \textbf{Adapted (40 imgs)}&$20.25$ & $0.00553$&& $0.01038$&$0.00215$\\
          \textbf{Adapted (60 imgs)} & $21.26$ &$0.00492$ &&$ 0.00827$ &$0.00193$\\
          \textbf{Adapted (80 imgs)}&\color{lightred}$\bf{21.67}$ &\color{lightred}$\bf{0.00469}$ &&\color{lightred}$\bf{0.00749}$ &\color{lightred}$\bf{0.00178}$\\
          \midrule
    \end{tabular}
    \label{tab:adaptation}
\end{table}

\begin{figure}[!t] 
	\centering 
	\begin{minipage}[t]{8cm} 
		\centering 
		\includegraphics[width=0.98\textwidth]{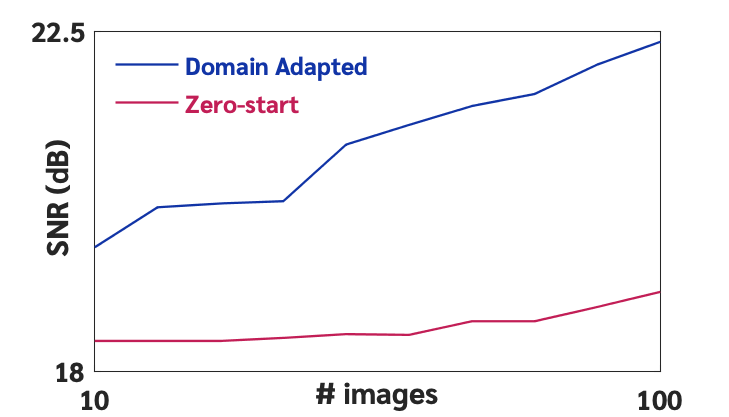} 
		\caption{Empirical evaluation of performance of PR-SANS with adapted and zero-start DL priors. Average SNR is plotted against number of samples used for adaptation and training of zero-start priors. Note the performance advantage of the domain adaptation  over zero-start DL models.}\label{fig:error}
	\end{minipage} 
	\hspace{0.5cm} 
	\begin{minipage}[t]{8cm} 
		\centering 
    \includegraphics[width=\textwidth]{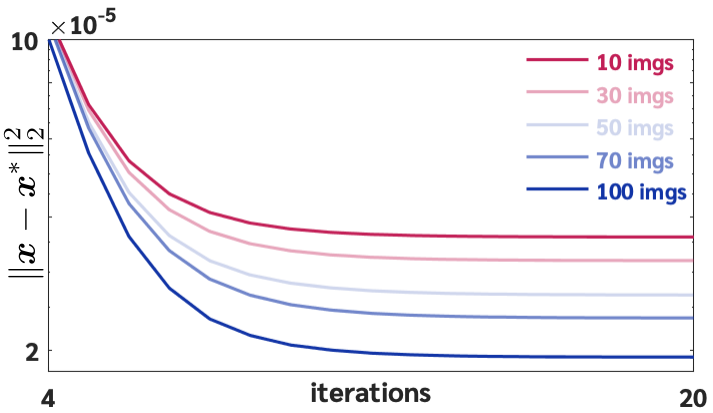}
    \caption{Empirical evaluation of PR-SANS convergence to the true solutions $\xbmast$ with multiple adapted priors. Note the error of restoration decrease with increase of images used for adaptation.}\label{fig:psnrofSamples}
	\end{minipage} 
\end{figure}

\section{Conclusion}
In this paper, we introduce PR-SANS, a PnP algorithm leveraging deep image priors  to recover high-SNR SANS data.
PR-SANS enhances the acquisition process by restoring high-SNR detector images from  low-SNR SANS data with shorter acquisition time. Our findings reveal that domain adaptation successfully addresses the challenge of limited training datasets, with a combination of 400 natural images and 100 SANS datasets proving adequate for obtaining adapted priors that successfully restore SANS detector images through PR-SANS. Comparative visual, empirical, and numerical analyses indicate that PR-SANS surpasses both traditional restoration algorithms, such as Total Variation (TV), and Convolutional Neural Networks (CNNs) in the restoration of SANS data.

\section*{Acknowledgment}

Neutron scattering research conducted using the Bio-SANS instrument, a DOE Office of Science, Office of Biological and Environmental Research resource (FWP ERKP291), used resources at the High-Flux Isotope Reactor, a DOE Office of Science, Scientific User Facility operated by the Oak Ridge National Laboratory.

\section{Appendix}\label{sec:app}
\subsection{Data collection procedure}\label{app:datacollection}
SANS measurements were performed at the CG3-Bio-SANS instrument at the High Flux Isotope Reactor (HFIR), at Oak Ridge National Laboratory (ORNL). The data was collected in a single instrument configuration with the main detector $15.5 ~m$  away from the sample; wing detector at $1.4^{\circ}$ rotational angle from the beam to cover a scattering vector ($Q$) range $0.003-0.8~\scriptstyle{\AA^{-1}}$ at a wavelength of  $6 \scriptstyle{\AA^{-1}} (\Delta \lambda / \lambda = 13.2\%)$; where $Q = (4\pi/\lambda) \sin\theta$,  $\theta$ is the scattering angle, and $\lambda$  is the neutron wavelength. The data were originally recorded for $1~hr$ to obtain a high-resolution 2D image. The recorded 2D scattering patterns were azimuthally averaged to obtain the high-resolution 1D SANS scattering curves $I(Q)$ versus scattering wave vector $Q$ using the ORNL developed software drtsans3~\cite{HELLER2022}. The high-resolution 2D image was sliced as a function of time  $5 ~min$ and azimuthally average to obtain low-resolution 1D data. The reduction process applied corrections such as dark current, transmission, and solid angle as well as subtracted parasitic scattering of the empty cell, and the final output data was in absolute scale. Details regarding the data collection and sample preparation processes are outlined in the appendices.

\noindent
\textbf{Materials.} Hybrid poplar (NM6) wood chips, D\textsubscript{2}O, and $\gamma$-Valerolactone-d\textsubscript{6}. Debarked Poplar wood chips were provided by the Great Lakes Bioenergy Research Center (GLBRC). The poplar was milled and sieved through a $5~mm$ screen before the lignin extraction. D\textsubscript{2}O was purchased from Cambridge Isotopes. d\textsubscript{6}-$\gamma$-Valerolactone (d\textsubscript{6}-GVL) was self-synthesized from levulinic acid (LA) $98\%$ (Sigma) following the reported method in~\cite{Yuan2021}.

\noindent
\textbf{Procedure for the $\mathbf{\gamma}$-Valerolactone-assisted lignin (GVL-lignin) isolation.} GVL-lignin was obtained by treatment of hybrid poplar wood chips using a published method~\cite{cheng2023}. GVL-lignin was extracted at $100$ and $120^{\circ}C$, referred to as $L100$ and $L120$. 

\noindent
\textbf{Procedure for the synthesis of d\textsubscript{6}-$\mathbf{\gamma}$-Valerolactone (d\textsubscript{6}-GVL).} [3,3,4,5,5,5-d]GVL (d\textsubscript{6}-GVL) was synthesized according to the reported methodology in~\cite{Yuan2021}. Briefly, $16.4~g$ of LA and $0.5~mL$ of H\textsubscript{2}SO\textsubscript{4} ($99.9\%$ from Fisher Scientific) were dissolved in $200~mL$ of D\textsubscript{2}O ($99.9$ atom $\%$ D from Aldrich), the mixture was heated at $90^\circ C$ and stirred for $72~hrs$ using a $450~mL$ Parr batch reactor. After $72~hrs$, the D\textsubscript{2}O was removed by rotary evaporation, $200~mL$ of new D\textsubscript{2}O was added to the reactor, and the deuteration continued for another $72~hrs$. This procedure was repeated twice for a total reaction time of $9$ days. The product obtained after the D\textsubscript{2}O removal was labeled as d\textsubscript{6}-LA and was analyzed by \textsuperscript{1}H and \textsuperscript{2}H-NMR to determine the deuterium incorporation. The d\textsubscript{6}-LA obtained was used for the synthesis of d\textsubscript{6}-GVL. $15.8~g$ of d\textsubscript{6}-LA and $0.13~g$ of $5$ wt$\%$ Ru/C were added in $30~g$ of D\textsubscript{2}O and heated at $90^\circ C$ under 40 bar of $25~mol\%$ of D\textsubscript{2}/N\textsubscript{2} (Airgas). The reaction mixture was stirred at $600~rpm$ for $3~hrs$ in a $75~mL$ Parr batch reactor and then quenched using an ice bath. 

\noindent
\textbf{SANS experiments.} A customized reaction cell was used for the experiments. $0.5$ and $10$ wt$\%$ lignin (L100/L120) were prepared in GVL-d\textsubscript{6}/D\textsubscript{2}O mixture ($90:10$ wt$\%$). The reaction cell was loaded with the pre-made lignin solution and placed in the beamline and heated to the target temperate ($25,~50,~80$, and back to $25^\circ C$) at a heating rate of $5^\circ C/min$. At each temperature stage, SANS data was collected for $1~hr$. 
\subsection{Proof of Theorem 1}
\begin{suptheorem}
\label{Thm:maintheoremapx}
Run PR-SANS with the adapted MMSE restoration operator $\Dhat_\sigma$ for $t\geq 1$ iterations under Assumptions~\ref{as:nondegen}-\ref{As:boundednessF} using a step-size $0 < \gamma \leq\min(1/M, 1/L)$. Then, for each iteration $1\leq k\leq t$, we have 
\begin{equation*}
    \min_{1\leq k \leq t} \|\nabla f(\xbm^k)\|^2_2 \leq \frac{1}{t}\sum_{k=1}^t \|\nabla f(\xbm^k)\|^2_2 
    \leq \frac{B_1}{t}\left(f(\xbm^0) - f(\xbm^*)\right) + B_2 \varepsilonbar_t^2, 
\end{equation*}
where $B_1 \defn 4(1+\gamma M)^2/(\gamma(1-\gamma M))$, $B_2 \defn \lambda \alpha^2 B_1/2 + 2\alpha^2(1+\gamma L)^2/\gamma^2$, and $\varepsilonbar_t^2 \defn  (1/t)\sum_{k=1}^t\varepsilon^2_{k-1}$ are iteration independent constants. Additionally, if the sequence of error terms $\{\varepsilon_i\}_{i \geq 0}$ is square-summable, we have that $\|\nabla f(\xbm^k)\| \to 0$ as $k\to \infty$.
\end{suptheorem}

\begin{proof}\label{Theorm1}
The update rule for adapted operator $ \Dsfhat_\sigma$ and ideal operator $\Dsf_\sigma$ are $\xbm^k = \prox_{\gamma g} \left( \xbm^{k-1} - \gamma \tau \Rsfhat(\xbm^{k-1})\right)$ and $\xbmbar^k = \prox_{\gamma g} \left( \xbm^{k-1} - \gamma \tau \Rsf(\xbm^{k-1})\right)$, respectively. 
From the optimality conditions for the update rule of ideal operator, we have 
\begin{equation*}
    \frac{1}{\gamma}(\xbmbar^k - \xbm^{k-1}) +\nabla h(\xbm^{k-1}) + \nabla g(\xbmbar^k) = 0.
\end{equation*}
By using this equation, the gradient of the objective function $f$ in~\eqref{eq:inverseproblem} can be written as 
\begin{align}\label{eq:gradobj}
   \nonumber\|\nabla f(\xbm^k)\|_2 &= \|\nabla g(\xbm^k) + \nabla h(\xbm^k)\|_2
     = \|\nabla g(\xbm^k) - \nabla g(\xbmbar^k) + \nabla h(\xbm^k) - \nabla h(\xbm^{k-1})+  \frac{1}{\gamma}(\xbm^{k-1} - \xbm^k) +  \frac{1}{\gamma}(\xbm^{k} - \xbmbar^k) \|_2 \\
    \nonumber &\leq \left(L+\frac{1}{\gamma}\right)\alpha \varepsilon_{k-1} + \left(M+\frac{1}{\gamma}\right) \|\xbm^k - \xbm^{k-1}\|_2, 
\end{align}
where we used Lemma~\ref{Lem:defineErrorTerm}, $L$-Lipschitz continuity of $\nabla g$, and $M$-Lipschitz continuity of $\nabla h$ from Assumptions~\ref{As:lipofG}. By squaring both sides and using $(a+b)^2 \leq 2a^2+2b^2$, we have 
\begin{equation}
\label{eq:finalgrad}
    \nonumber\|\nabla f(\xbm^k)\|^2_2 \leq A_1\|\xbm^k - \xbm^{k-1}\|_2^2 + A_2\varepsilon^2_{k-1}, 
\end{equation}
where $A_1 \defn 2(M + 1/\gamma)^2$ and $A_2 \defn 2\alpha^2(L + 1/\gamma)^2 $. By using the result from Lemma~\ref{Lem:DecreasingF} and averaging both sides of the bound over $t\geq 1$, we get the desired result
\begin{equation*}
    \min_{1\leq k \leq t} \|\nabla f(\xbm^k)\|^2_2 \leq \frac{1}{t}\sum_{k=1}^t \|\nabla f(\xbm^k)\|^2_2 
    \leq \frac{B_1}{t}\left(f(\xbm^0) - f(\xbm^t)\right) + \frac{B_2}{t} \sum_{k=1}^t\varepsilon^2_{k-1}
    \leq \frac{B_1}{t}\left(f(\xbm^0) - f(\xbm^*)\right) + B_2 \varepsilonbar_t^2
\end{equation*}
where we used the fact that $f(\xbm^t) \geq f(\xbmast)$ from Assumption~\ref{As:boundednessF}, and $B_1 \defn 4(1+\gamma M)^2/(\gamma(1-\gamma M))$, $B_2 \defn \lambda \alpha^2 B_1/2 + 2\alpha^2(1+\gamma L)^2/\gamma^2$, and $\varepsilonbar_t^2 \defn  (1/t)\sum_{k=1}^t\varepsilon^2_{k-1}$.
\end{proof}
\noindent\textbf{Remark 1.} 
If the sequence of error terms $\{\varepsilon_i\}_{i \geq 0}$ is square-summable, we have $\varepsilonbar_t^2 \to 0$ as $t \to \infty$. Consequently, $\|\nabla f(\xbm^t)\| \to 0$ as $t \to \infty$.

\begin{lemma}
    \label{Lem:DecreasingF}
    Run PR-SANS with adapted restoration operator $\Dhat_\sigma$ for $k\geq 1$ iterations under Assumptions~\ref{as:nondegen}-\ref{As:boundednessF} using a step-size $0 < \gamma \leq\min(1/M, 1/L)$. Then we have 
\begin{equation*}
  f(\xbm^k)\leq f(\xbm^{k-1})  - \frac{1- \gamma M}{2\gamma}\|\xbm^k - \xbm^{k-1}\|_2^2 +  \frac{\lambda \alpha^2 \varepsilon^2_{k-1}}{2}. 
\end{equation*}
\end{lemma}

\begin{proof}
    Consider the iteration $k\geq 1$ of the algorithm 
\begin{equation*}
    \xbm^k = \prox_{\gamma g} \left( \xbm^{k-1} - \gamma \tau \Rsfhat(\xbm^{k-1})\right)~ \text{with} ~ \Rsfhat\defn \Isf - \Dhat_\sigma, 
\end{equation*}
$\Dhat_\sigma$ is the adapted restoration operator.  Since $\Dsf_\sigma$ is the ideal MMSE restoration operator specified in~\eqref{eq:mmsedenoiser}, $\xbmbar^k = \prox_{\gamma g} \left( \xbm^{k-1} - \gamma \tau \Rsf(\xbm^{k-1})\right)$ minimizes 
\begin{equation*}
    \phi(\xbm) \defn \frac{1}{2\gamma}\|\xbm - \left(\xbm^{k-1} - \tau \gamma \left(\xbm^{k-1} - \Dsf_\sigma(\xbm^{k-1})\right)\right)\|_2^2+ g(\xbm)
     = \frac{1}{2\gamma}\|\xbm - \left(\xbm^{k-1} -  \gamma \nabla h(\xbm^{k-1}) \right)\|_2^2 +g(\xbm), 
\end{equation*}
where we used the result of Lemma~\ref{Lem:regularizerDenRel}. From Assumption~\ref{As:lipofG}, we know that $\nabla g$ is $L-$Lipschitz continuous, which implies
\begin{equation*}
    \|\nabla \phi(\vbm) - \nabla \phi(\ubm)\|_2 \leq \lambda\|\vbm - \ubm\|_2~\text{with}~\lambda\defn\frac{1}{\gamma} +L.
\end{equation*}
By using the fact that $\phi$ is a smooth function from Assumption~\ref{As:lipofG} and since $\xbmbar^k$ minimize it, we have 
\begin{align*}
    \phi(\xbm^k) &\leq \phi(\xbmbar^k) + \nabla \phi(\xbmbar^k)^\Tsf(\xbm^k - \xbmbar^k)+ \frac{\lambda}{2}\|\xbm^k - \xbmbar^k\|_2^2 \leq \phi(\xbmbar^k) + \frac{\lambda\alpha^2 \varepsilon^2_{k-1}}{2}\\
   & \leq \min_{\xbm}\left\{\frac{1}{2\gamma}\|\xbm- \left(\xbm^{k-1} -  \gamma \nabla h(\xbm^{k-1})\right)\|_2^2 + g(\xbm) \right\}  + \frac{\lambda\alpha^2\varepsilon^2_{k-1}}{2}\\
   & \leq \frac{1}{2\gamma}\|\xbm^{k-1} - \left(\xbm^{k-1} -  \gamma \nabla h(\xbm^{k-1})\right)\|_2^2 + g(\xbm^{k-1})  + \frac{\lambda \alpha^2\varepsilon^2_{k-1}}{2}, 
\end{align*}
where we used the result of Lemma~\ref{Lem:defineErrorTerm} in the first inequality. 
By expanding the first term on the left side of the inequality and simplifying it, we have
\begin{equation}
    \label{eq:inequalityg}
    \nonumber g(\xbm^k) \leq g(\xbm^{k-1}) - \nabla h(\xbm^{k-1})^\Tsf(\xbm^k - \xbm^{k-1})
    -\frac{1}{2\gamma}\|\xbm^k - \xbm^{k-1}\|_2^2+  \frac{\lambda \alpha^2 \varepsilon^2_{k-1}}{2}. 
\end{equation}
From the $M-$Lipschitz continuity of $\nabla h$, we have 
\begin{equation}\label{eq:inequlityH}
     h(\xbm^k) \leq h(\xbm^{k-1}) + \nabla h(\xbm^{k-1})^\Tsf(\xbm^k - \xbm^{k-1})
    +\frac{M}{2}\|\xbm^k - \xbm^{k-1}\|_2^2. 
\end{equation}
By combining~\eqref{eq:inequalityg} and~\eqref{eq:inequlityH}, we have 
\begin{align}\label{eq:decf}
    \nonumber f(\xbm^k) = g(\xbm^k) +h(\xbm^k) &\leq g(\xbm^{k-1}) +h(\xbm^{k-1}) - \frac{1- \gamma M}{2\gamma}\|\xbm^k - \xbm^{k-1}\|_2^2 +  \frac{\lambda \alpha^2 \varepsilon^2_{k-1}}{2}\\
    & \leq f(\xbm^{k-1})  - \frac{1- \gamma M}{2\gamma}\|\xbm^k - \xbm^{k-1}\|_2^2 + \frac{\lambda \alpha^2 \varepsilon^2_{k-1}}{2}. 
\end{align}
\end{proof}

\begin{lemma}\label{Lem:defineErrorTerm}
Run PR-SANS with adapted restoration operator $\Dhat_\sigma$ for $k\geq 1$ iterations under Assumptions~\ref{as:nondegen}-\ref{As:boundednessF} using a step-size $0 < \gamma \leq\min(1/M, 1/L)$. Then we have 
\begin{equation*}
\|\xbm^k - \xbmbar^k\|_2 \leq \alpha \varepsilon_{k-1}, 
\end{equation*}
where $\alpha \defn \gamma \tau /(1 - \gamma \tau )$ is a positive constant. 
\end{lemma}
\begin{proof}
By using the ideal $\Dsf_\sigma$ and adapted $\Dsfhat_\sigma$ restoration operators, we have the following expressions:
\begin{equation*}
    \xbmbar^k = \prox_{\gamma g }  ( \xbm^{k-1} - \gamma \tau \Rsf(\xbm^{k-1})) \quad\text{and}\quad 
   \xbm^k = \prox_{\gamma g } (\xbm^{k-1} - \gamma \tau \Rsfhat(\xbm^{k-1})). 
\end{equation*}
By using these expressions, the optimality conditions for proximal operators and Assumption~\ref{As:InexactDistance}, we have 
\begin{equation}
\label{eq:boundonDenXrelation}
\|\xbm^k - \xbmbar^k + \gamma (\nabla g(\xbm^k) - \nabla g(\xbmbar^k))\|_2 = \gamma \tau \|\Rsf(\xbm^{k-1}) - \Rsfhat(\xbm^{k-1})\|_2= \gamma \tau\|\Dsf_\sigma(\xbm^{k-1}) - \Dsfhat_\sigma(\xbm^{k-1})\|_2 \leq \gamma \tau\varepsilon_{k-1}. 
\end{equation}
By this equality and $\|a\|_2   - \|b\|_2 \leq \|a+b\|_2$ and $L$-Lipschitz continuity of  $\nabla g$, we have 
\begin{equation}
\label{eq:boundonX}
 \|\xbm^k - \xbmbar^k\|_2 -  \gamma \|(\nabla g(\xbm^k) - \nabla g(\xbmbar^k))\|_2\geq 
    (1-\gamma L) \|\xbm^k - \xbmbar^k\|_2.  
\end{equation}
From equations~\eqref{eq:boundonDenXrelation} and~\eqref{eq:boundonX} and the fact that $\gamma L \leq 1$, we have 
\begin{equation*}
    \|\xbm^k - \xbmbar^k\|_2 \leq \alpha \varepsilon_{k-1}, 
\end{equation*}
where $\alpha \defn \gamma \tau /(1- \gamma L)$. 
\end{proof}

\begin{lemma}\label{Lem:regularizerDenRel}
Let $\Dsf_\sigma$ be the ideal MMSE restoration operator in~\ref{eq:mmsedenoiser} corresponding to restoration problem~\ref{eq:noisemodel} under Assumption~\ref{as:nondegen} and~\ref{As:lipofG}. Then we have
$$
\nabla_\zbm h(\zbm) =\tau (\zbm - \Dsf_\sigma(\zbm)), 
$$
where $h$ is the regularization term in~\ref{eq:inverseproblem}.
\end{lemma}

\begin{proof}
The density $p_{\zbm}$ is defined as 
\begin{equation}
\label{eq:densPrior}
   \nonumber p_\zbm(\zbm) = \int p_{\zbm|\xbm}(\zbm;\xbm) p_\xbm (\xbm)\dsf \xbm = \int G_\sigma(\zbm - \xbm) p_\xbm (\xbm)\dsf \xbm, 
\end{equation}
where $G_\sigma$ denotes the Gaussian density function with the standard deviation $\sigma >0$. 
The gradient of $h(\zbm) = -\tau \sigma^2 \log p_{\zbm}(\zbm)$ can be expressed as 
\begin{equation}\label{eq:prior}
    \nonumber\nabla_\zbm h(\zbm) = \frac{-\tau \sigma^2 \nabla_\zbm p_\zbm(\zbm)}{p_\zbm(\zbm)}= \frac{\tau  \int (\zbm - \xbm)G_\sigma(\zbm - \xbm) p_\xbm (\xbm)\dsf \xbm}{p_\zbm(\zbm)} =\tau (\zbm - \Dsf_\sigma(\zbm)). 
\end{equation}
where we used~\eqref{eq:regden}, ~\eqref{eq:densPrior},and~\eqref{eq:mmsedenoiser} in the first, second and, last equality, respectively. 
\end{proof}
\noindent\textbf{Remark 2.} 
For the fixed points of PR-SANS algorithm using ideal operator $\Dsf_\sigma$, we have 
$$
\xbmast = \prox_{\gamma g} (\xbmast - \gamma \tau \Rsf(\xbmast)). 
$$
By using the results from Lemma~\ref{Lem:regularizerDenRel} and the optimality condition for proximal operators, we have  
\begin{equation*}
 \frac{1}{\gamma} \left(\xbmast - \xbmast + \gamma \tau \Rsf(\xbmast) \right) + \nabla g(\xbmast) = \nabla g(\xbmast) + \tau \Rsf(\xbmast)=  \nabla g(\xbmast) + \nabla h(\xbmast)  =\bm{0},
\end{equation*}
where $\Rsf = \Isf -\Dsf$. The result states that any fixed-point $\xbmast\in \R^n$  of the PR-SANS with ideal restoration operator is indeed a stationary point of objective function $f$. 
\bibliographystyle{IEEEbib}
\bibliography{refs}
\end{document}